\documentclass[10pt]{book}
\usepackage{array,epsfig,fancyheadings,rotating}

\usepackage{bm,booktabs,color}
\usepackage{amsmath}
\usepackage{amssymb}
\usepackage{amsfonts}
\usepackage{multirow}
\usepackage{amsthm}

\newtheorem{theorem}{Theorem}
\newtheorem{lemma}{Lemma}

\theoremstyle{definition}

\newtheorem{remark}{Remark}
\newcommand{\Sigmab}{{\boldsymbol\Sigma}}
\newcommand{\I}{{\boldsymbol I}}
\newcommand{\0}{{\boldsymbol 0}}

\newcommand{\A}{{\boldsymbol A}}
\newcommand{\B}{{\boldsymbol B}}
\newcommand{\Cov}{{\mbox{Cov}}}

\newcommand{\itemEq}[1]{%
        \begingroup%
        \setlength{\abovedisplayskip}{0pt}%
        \setlength{\belowdisplayskip}{0pt}%
        \parbox[c]{\linewidth}{\begin{flalign}#1&&\end{flalign}}%
        \endgroup}

\newcommand{\tr}{\mbox{tr}}
\begin{document}


\renewcommand{\baselinestretch}{1.2}

\markright{ \hbox{\footnotesize\rm 
}\hfill\\[-13pt]
\hbox{\footnotesize\rm
}\hfill }

\markboth{\hfill{\footnotesize\rm Tung-Lung Wu AND Ping Li} \hfill}
{\hfill {\footnotesize\rm Tests for High-Dimensional Covariance Matrices Using Random Matrix Projection} \hfill}

\renewcommand{\thefootnote}{}
$\ $\par


\fontsize{10.95}{14pt plus.8pt minus .6pt}\selectfont
\vspace{0.8pc}
\centerline{\LARGE\bf Tests for High-Dimensional Covariance Matrices}
 \vspace{2pt}
 \centerline{\LARGE\bf  Using Random Matrix Projection	}
\vspace{.4cm}
\centerline{Tung-Lung Wu and Ping Li}
\vspace{.4cm}
\centerline{\it Department of Statistics and Biostatistics, Department of Computer Science, Rutgers University}
\vspace{.55cm}
\fontsize{9}{11.5pt plus.8pt minus .6pt}\selectfont


\begin{quotation}
\noindent {\it Abstract:}
The classic likelihood ratio test for testing the equality of two  covariance matrices breakdowns due to the singularity of the sample covariance matrices when the data dimension $p$ is larger than the sample size $n$. In this paper, we present a conceptually simple method using random projection to project the data onto the one-dimensional random subspace so that the conventional methods can be applied.  Both one-sample and two-sample tests for high-dimensional covariance matrices are studied.  Asymptotic results are established and numerical results are given to compare our method with state-of-the-art methods in the literature.  \par


\vspace{9pt}
\noindent {\it Key words and phrases:} Covariance matrices, High-dimensional, likelihood ratio test, random projection, subspace, large $p$ small $n$.
\par
\end{quotation}\par

\def\thefigure{\arabic{figure}}
\def\thetable{\arabic{table}}

\fontsize{10.95}{14pt plus.8pt minus .6pt}\selectfont

\setcounter{chapter}{1}
\setcounter{equation}{0} 
\noindent {\bf 1. Introduction}

One-sample and two-sample testing problems for high-dimensional covariance matrices are considered in this paper. In high-dimensional setting, the conventional methods fail usually due to the singularity of the sample covariance matrices. Consider one-sample test and let $X_1,\ldots,X_n$ follow a $p$-dimensional normal distribution $N_p(\0,\Sigmab)$. We want to test
\begin{equation}\label{t-prob1}
 H^1_0: \Sigmab = \I,
\end{equation}
where $\I$ is a $p\times p$ identity matrix.  Note that for a given covariance matrix $\Sigmab_0$ we can always test (\ref{t-prob1}) based on the transformed data $\tilde{X}_k = \Sigmab_0^{-1/2}X_k$.

The likelihood ratio test statistic for (\ref{t-prob1}) is given by
\begin{align}\label{lrt}
T_1 = n\cdot(\tr(S) - \log|S| -p),
\end{align}
where $S = \sum^{n}_{k=1}X_kX_k^{\intercal}/n$ is the sample covariance matrix and $\tr(S)$ denotes the trace of $S$. The likelihood ratio test performs poorly when $p$ increases as $n$ tends to infinity.  It has been shown  numerically that the size of the test based on (\ref{lrt}) is 100\% in the case $(p,n)=(300,500)$ (\cite{Bai-et-al-2009}). Further, the test statistic is undefined when $p>n$ due to the singularity of the  sample covariance matrix. Bai {\em et. al.}~\cite{Bai-et-al-2009} proposed a corrected likelihood ratio test (CLRT) with a condition $p/n\to c\in (0,1).$ They established the asymptotic normality result for a corrected version of $T_1$ using random matrix theory. Some related  works on CLRT can be found in \cite{Jiang-2012} and \cite{Jiang-2013}. Instead of  using the sample covariance matrix, Chen {\em et. al.}~\cite{chen-2010} proposed an one-sample test based on more accurate estimators of $\tr(\Sigmab)$ and $\tr({\Sigmab^2})$ with the assumption $\tr({\Sigmab^4})=o(\tr^2({\Sigmab^2}))$.

In two-sample test, let $X_1,\ldots,X_{n_1}$ follow a $p$-dimensional normal distribution $N_p(\0,\Sigmab_1)$ and $Y_1,\ldots,Y_{n_2}$ follow a $p$-dimensional normal distribution $N_p(\0,\Sigmab_2)$. We want to test
\begin{equation}
 H^2_0: \Sigmab_1 = \Sigmab_2.
\end{equation}
The likelihood ratio test statistic
\begin{equation}
T_2 = -2 \log \frac{|S_1|^{n_1/2} \cdot |S_2|^{n_2/2}}{|c_1S_1+c_2S_2|^{(n_1+n_2)/2}},
\end{equation}
where $S_1$ and $S_2$ are the sample covariance matrices of  $\{X_k\}^{n_1}_{k=1}$ and $\{Y_k\}^{n_2}_{k=1}$, respectively, and $c_j=n_j/(n_1+n_2)$, $j=1,2$, encounters the same problem that the sample covariance matrices are singular when $p>n$.

There have been advances in the field of testing high-dimensional covariance matrices. We recognize three approaches in this field. The limiting distribution of extreme eigenvalues of the sample covariance matrix is derived in \cite{Bai-1993} and \cite{Bai-1993a} based on random matrix theory.  Bickel and Levina~\cite{Levina-2008b,Levina-2008a} and Li and Chen~\cite{li-2012} derive consistent and better estimators to the population covariance matrices.  The regularized covariance estimator is proposed by solving  a maximum likelihood estimation problem subject to a constrain on the condition number~\cite{bala-2013}.
There is the line of work of using random projections for testing two-sample means in high-dimension~\cite{Lopes-2014,Li-2015}.

 In this paper, we study the random projection method with focus on projecting the data onto only one-dimensional subspace. Therefore, any one-dimensional test can be used on the projected data. Surprisingly, the one-dimensional random projection turns out to be quite remarkable on certain class of covariance matrices. The foundation of random projection method is the lemma in \cite{Johnson-1984} where the distances between projected data points are approximately preserved.
The reason for adopting the method is
three folds: (i) conceptually simple, (ii) easy to program and (iii) efficient in computation. We will illustrate the method based on some conventional statistics. We want to emphasize that by reducing the dimension using random projection, we do not require any explicit relationship between $p$ and $n$,  unless otherwise mentioned.

This rest of the paper is organized in the following way. In Section 2, the one-sample test is considered. In Section 3, the two-sample test is considered.  Numerical results to compare our method with two other well-known methods are given in Section 4.  Summary and discussion are given in Section 5.

\par

\vspace{0.1in}

\setcounter{chapter}{2}
\setcounter{equation}{0} 
\noindent {\bf 2. One-sample tests}

Let $X_1,\ldots,X_n$ follow a $p$-dimensional normal distribution $N_p(\0,\Sigmab)$. We want to test
\begin{equation}\label{t-prob3}
 H^1_0: \Sigmab = \I.
\end{equation}

Given a $p\times 1$ random projection vector $R$, the projected data is given by
\begin{equation}\label{yk}
Y_k= R^{\intercal}X_k, \quad k=1,\ldots,n,
\end{equation}
where $Y_k$, conditioned on $R$, follows $N(0,\sigma^2=R^{\intercal}\Sigmab R)$. Note that for notational simplicity the subscript $p$ for one-dimensional normal distribution is suppressed.   Before we continue to develop any tests for the one-dimensional projected data, we briefly discuss about the choice of the random projection vector $R$. In practice, we may use any random projection $R^{\intercal}=(r_1,\ldots,r_p)$ with mean $\0$ and $E(r_i r_j) = 0$, $i\not =j$. A particular choice is the normalized  random vector with independent $N(0,1)$ entries. The advantage of using such a random projection is three folds: (i) it is easy to interpret, (ii) the convergence rate is fast and (iii) the random vector is well studied~\cite{fang-1990}. The reasoning is as follows. Given a vector $R$, the problem is reduced to test if the variance  of the projected data, $\sigma^2$, is equal to 1, since $R^{\intercal}\I R = 1$ under the null hypothesis. Simulations under different choices of random projection vectors were conducted and the results have indicated that the convergence rate of the asymptotic normality based on the normalized version is better than the one based on the  non-normalized random vector.

Note that, for computational efficiency, Srivastava, Li and Ruppert~\cite{Li-2015} also proposed the use of ``one permutation + one random projection'', a variant which borrowed the idea from ``very sparse random projections'' in~\cite{li-2006} and ``one permutation hashing'' in~\cite{LiOwenZhang-2012}.

The danger of using only one projection is that the conclusion may be completely opposite (low power) for the same data set using different projections. The following example is  unrealistic but clearly illustrates the problem. Let the covariance matrix be
\[
\Sigmab = \begin{pmatrix}
  3 & 2.5 & 2 \\
  2.5 & 2 & 1.5 \\
  2 & 1.5 & 1
 \end{pmatrix},
\]
and the projection vector is $R^{\intercal}=(0,0,1)$. Rejecting the null hypothesis is unlikely since $\sigma^2 = R^{\intercal}\Sigmab R =1.$ One solution is to use $m$ random projections as follows:
\begin{enumerate}
	\item Sample $m$ independent random vectors, $R_1,\ldots,R_m$.
	
	\item For each $R_i, i=1,\ldots,m$, compute the projected data $Y^{i}_k$, $k=1,\ldots,n$ and the statistic $T^i_n$.

	\item The maximum value $T_{1,n} = \max_{1\leq i \leq m}T^i_n$ is used as test statistic for (\ref{t-prob3}).
\end{enumerate}

In the sequel, we develop the test statistic and the asymptotic properties for one-sample test. Given the data and the random projection $R_i$, the conventional statistic
\begin{equation}\label{test}
	  \sum^{n}_{k=1} {Y^i_k}^2
\end{equation}
is used, where $Y^i_k$ is given in~(\ref{yk}).
The statistic in (\ref{test}) is sufficient and chi-square distributed with $n$ degrees of freedom.  It follows that the standardized statistic  $\tilde{T}^i_n = \frac{\sum^{n}_{k=1} {Y^i_k}^2  -n}{\sqrt{2n}}$ converges in distribution to the standard normal $N(0,1)$. I.e.,
\begin{equation}\label{m-test1}
	P(\tilde{T}^i_n<x) - \Phi(x) = o(1).
\end{equation}
However, it is a known fact that the convergence  of $(\ref{m-test1})$ is slow. Hence, here we apply the square root transformation
\begin{equation}\label{mo-test}
	{T^i_n} = \sqrt{2\sum^{n}_{k=1} {Y^i_k}^2 }-\sqrt{2n-1}.
\end{equation}

The following lemma is a well-known result due to \cite{fisher-1973}.
\begin{lemma}\label{fisher}
As $n\to\infty$,
$$T^i_n= \sqrt{2\sum^{n}_{k=1} {Y^i_k}^2 }-\sqrt{2n-1} \overset{D}{\to}  N(0,1),$$
where $\overset{D}{\to} $ denotes the convergence in distribution.
\end{lemma}
To develop the asymptotic distribution of the test statistic, we need some quadratic form results. Assume $x\sim N(\0,\Sigmab)$ and $\A$ and $\B$ to be symmetric matrices. Then
\begin{itemize}
\item
\itemEq{ \label{r1} E(x^{\intercal}\A x) = \mbox{\tr}(\A\Sigmab).}

\item
\itemEq{ \label{r2}  \mbox{Cov}(x^{\intercal}\A x,x^{\intercal}\B x) = 2\mbox{\tr}(\A\Sigmab\B\Sigmab).}

\end{itemize}
Note that equation~(\ref{r1}) does not require normality.

\begin{lemma}\label{cov}
Under $H^1_0$, for any $k\in\{1,\ldots,n\}$ and $i\not= j\in\{1,\ldots,m\}$,
\begin{equation}
\Cov(X^{\intercal}_kR_iR_i^{\intercal}X_k, X^{\intercal}_kR_jR_j^{\intercal}X_k) = 2/p.
\end{equation}

\end{lemma}

\begin{proof}
It follows from (\ref{r1}) and (\ref{r2}) that
\begin{align*}
\Cov({Y^{i}_k}^2,{Y^{j}_k}^2) & = E(\Cov({Y^{i}_k}^2,{Y^{j}_k}^2|R_i,R_j)) + \Cov(E({Y^{i}_k}^2|R_i),E({Y^{j}_k}^2|R_j))\\
       & = E(\Cov(X_k^{\intercal}R_iR_i^{\intercal}X_k,X_k^{\intercal}R_jR_j^{\intercal}X_k|R_i,R_j)) \\
       & = 2 E(\tr(R_iR_i^{\intercal}R_jR_j^{\intercal}))\\
       & = 2 E(E(R_j^{\intercal}R_iR_i^{\intercal}R_j)|R_i)\\
       & =  2 E\left(\tr\left(R_iR_i^{\intercal}\frac{1}{p}\I\right)\right)\\
       & = \frac{2}{p}E(R_i^{\intercal}R_i)\\
       & = \frac{2}{p}.
\end{align*}
The proof is completed.
\end{proof}

\begin{lemma}\label{asym-ind}
For $i\not = j$, ${\tilde{T}^i_n}$ and ${\tilde{T}^j_n}$ are asymptotically independent.
\end{lemma}
\begin{proof}
It is easy to see that ${\tilde{T}^i_n}$ and ${\tilde{T}^j_n}$ are asymptotically normally distributed with standard deviation 1. We now show the covariance between ${\tilde{T}^i_n}$ and ${\tilde{T}^j_n}$ is 0 as $p\to\infty$.
It follows from Lemma~\ref{cov} that
 \begin{align*}
 \Cov({\tilde{T}^i_n},{\tilde{T}^j_n}) &= \frac{1}{2n}  \Cov\left(\sum_k X^{\intercal}_kR_iR_i^{\intercal}X_k,\sum_k X^{\intercal}_kR_jR_j^{\intercal}X_k\right) \\
 						         &= \frac{1}{2}\Cov( X^{\intercal}_kR_iR_i^{\intercal}X_k, X^{\intercal}_kR_jR_j^{\intercal}X_k) \\
 						         &= 1/p \to 0, \quad \mbox{as } p\to \infty.
 \end{align*}
The covariance is independent of $n$ and this completes the proof.
\end{proof}

\begin{theorem}\label{thm1}
Under $H^1_0,$
\begin{equation}
T_{1,n}  \overset{D}{\to}  \max_{1\leq i \leq m}Z_i,
\end{equation}
where  $Z_i's$ are independent standard normal.
\end{theorem}

\begin{proof}
Based on multivariate central limit theorem~\cite{van-1996}, it follows from (\ref{test}) and (\ref{mo-test}) that we have
\begin{equation}
({T^1_n},\ldots,{T^m_n}) \overset{D}{\to} N(\0,\Sigma_{T}),
\end{equation}
where $(\Sigma_{T})_{ij} =\mbox{Cov}({T^i_n},{T^j_n})$.
It follows from Lemma~\ref{asym-ind} that the covariance matrix is approaching identity matrix as $p\to \infty$. I.e., they are asymptotically independent.
 This completes the proof.
\end{proof}

\begin{remark}
To account for finite $p$, the exact covariance matrix can be obtained numerically via (5.5.7) in  \cite{Harvey-1965}.
\end{remark}

\begin{remark}
If we want to test $H^1_0: \Sigmab=\I $ against $H^1_a:  \Sigmab-\I$  is a positive (negative) definite matrix, the test statistic $ \max_{1\leq i \leq m}T^i_n$  ($ \min_{1\leq i \leq m}T^i_n$) is a natural choice.
For general alternative hypothesis, we may use the modified test statistic $(\max_{1\leq i \leq m}T^i_n,\min_{1\leq i \leq m}T^i_n$) and reject the null hypothesis if
 $\max_{1\leq i \leq m}T^i_n \geq c_{\max}$ or $\min_{1\leq i \leq m}T^i_n \leq c_{\min}$, where $c_{\max}$ and $c_{\min}$ satisfy
 \begin{equation}\label{two-sided}
 P\left(\max_{1\leq i \leq m}T^i_n>c_{\max} \mbox{ or } \min_{1\leq i \leq m}T^i_n \leq c_{\min}\right) = \alpha,
 \end{equation}
 and can be chosen to be $z_{1-(1-\alpha/2)^{1/m}}$ and $-z_{1-(1-\alpha/2)^{1/m}}$, respectively, at a given $\alpha$ level. We denote by $z_{\alpha}$ the upper $\alpha\times100$ percentile of the  standard normal distribution.   The test using (\ref{two-sided}) is referred to as two-sided test, while the test using  $ \max_{1\leq i \leq m}T^i_n$  (or $ \min_{1\leq i \leq m}T^i_n$) is referred to as one-sided test.
\end{remark}

\setcounter{chapter}{2}
\noindent {\bf 2.1 Approximation of covariance between different random projections}

Based on  Lemma~\ref{asym-ind}, we understand that the covariance between different random projections tends to 0 as $p\to \infty$. Knowing the convergence rate would give more insight  into our theorems. With additional assumptions given below, we can approximate the convergence rate of the covariance matrix $\Sigmab_T$ using method of moments.

\begin{description}
\item[Assumption 1]
Let $\lambda_i,i=1,\ldots,p,$ be eigenvalues of $\frac{1}{n}\sum^{n}_{k=1} X_kX_k^{\intercal}$. The average $\frac{1}{p}\sum^{p}_{i=1} \lambda_i$ of eigenvalues  is uniformly integrable.

\item[Assumption 2]
Let $n\to \infty$ and $p\to\infty$ in such a way that $\frac{p}{n} \to y$, $0\leq y<\infty.$

\end{description}


It follows from the definition of covariance that
\begin{align*}
\Cov({T^i_n},{T^j_n})  & ={2}\Cov\left(\left(\sum^{n}_{k=1} {{Y^i_k}^2}\right)^{1/2},\left(\sum^{n}_{k=1} {{Y^j_k}^2}\right)^{1/2}\right)\\
					    &  =  {2}\Cov\left( \left(\sum^{n}_{k=1}R^{\intercal}_iX_kX_k^{\intercal}R_i\right)^{1/2}, \left(\sum^{n}_{k=1}R^{\intercal}_jX_kX_k^{\intercal}R_j\right)^{1/2}\right)             \\
                                           & ={{2}}\Bigg\{  E\left(\left( \sum^{n}_{k=1}R^{\intercal}_iX_kX_k^{\intercal}R_i\right)^{1/2}\left( \sum^{n}_{k=1}R^{\intercal}_jX_kX_k^{\intercal}R_j\right)^{1/2}\right) \\
 & \quad-E \left( \sum^{n}_{k=1}R^{\intercal}_iX_kX_k^{\intercal}R_i\right)^{1/2} E\left( \sum^{n}_{k=1}R^{\intercal}_jX_kX_k^{\intercal}R_j\right)^{1/2} \Bigg\}               \\
   & = {{2}}(A-B),
\end{align*}
and
\begin{align*}
A = E\left\{E\left( \left(\sum^{n}_{k=1}R^{\intercal}_iX_kX_k^{\intercal}R_i\right)^{1/2}\bigg|\underline{\boldsymbol X}\right)E\left(\left( \sum^{n}_{k=1}R^{\intercal}_jX_kX_k^{\intercal}R_j\right)^{1/2}\bigg|\underline{\boldsymbol X}\right)\right\},
\end{align*}
where $\underline{\boldsymbol X}=(X_1,\ldots,X_n)$.

Let $S_n=\frac{1}{n}\sum^{n}_{k=1}X_kX_k^{\intercal}$ and $\lambda_i,i=1,\ldots,p,$ be the eigenvalues of $S_n$ and $k_1= \sum^{p}_{i=1} \lambda_i$ and $k_2 =2\sum^{p}_{i=1}\lambda_i^2$.
Using Patnaik's approximation~\cite{Harvey-1965} by matching moments, the moments of quadratic forms can be approximated by the moments of a scaled chi-squared distribution:
\begin{align*}
E\left( \left(\sum^{n}_{k=1}R^{\intercal}_iX_kX_k^{\intercal}R_i\right)^{1/2}\bigg|\underline{\boldsymbol X}\right) \approx E(Y_1^{1/2}),
\end{align*}
where $Y_1 \sim a_1\cdot \chi^2_{\nu_1}$ and $a_1 = \frac{k_2}{2k_1}$, $\nu_1=\frac{2k_1^2}{k_2}$. One can show that
\begin{equation}\label{approx-m}
E(Y_1^{1/2}) = \left( \frac{k_2}{k_1}\right)^{1/2}  \frac{\Gamma\left( \frac{k_1^2}{k_2}+1/2\right)}{\Gamma\left( \frac{k_1^2}{k_2}\right)}.
\end{equation}
 It follows from \cite{Jonsson-1982} that $k_1/p \to 1$ and $k_2/p \to 2(1+y)$ in probability.   By the asymptotic representation of gamma function, the following ratio of gamma functions can be approximated by
\begin{equation}\label{expan}
\Gamma(a+1/2)/\Gamma(a) = \left(a^{1/2} -\frac{1}{8}a^{-1/2}\right) (1+O(a^{-3/2})).
\end{equation}

Substituting (\ref{expan}) into (\ref{approx-m}) and using Assumption 1, we can interchange the limit and expectation and obtain
\begin{equation}
A \approx p-\frac{1+y}{2} +O(p^{-1/2}) \mbox{ and } B \approx p-\frac{1+y}{2} +O(p^{-1/2}).
\end{equation}
Therefore, $\Cov({T^i_n},{T^j_n}) \approx O(p^{-1/2})$.



\lhead[\footnotesize\thepage\fancyplain{}\leftmark]{}\rhead[]{\fancyplain{}\rightmark\footnotesize\thepage}

\vspace{0.1in}

\setcounter{chapter}{3}
\setcounter{equation}{0} 
\noindent {\bf 3. Two-sample tests}

Let $X_1,\ldots,X_{n_1}$ follow a $p$-dimensional normal distribution $N(\0,\Sigmab_1)$ and $Y_1,\ldots,Y_{n_2}$ follow a $p$-dimensional normal distribution $N(\0,\Sigmab_2)$. Let $S_1$ and $S_2$ be the sample covariance matrices calculated from $\{X_k\}^{n_1}_{k=1}$ and $\{Y_k\}^{n_2}_{k=1}$, respectively. We want to test
\begin{equation}\label{2t-prob}
 H^2_0: \Sigmab_1 = \Sigmab_2.
\end{equation}

We project the two samples $X_k$ and $Y_k$ using normalized Gaussian random vectors $R_i, i=1,\ldots,m$. Thus, given $R_i$, the projected data	
$X_k^i = R_i^{\intercal}X_k \sim N(\0,R_i^{\intercal}\Sigmab_1R_i = \sigma^2_1)$
and
$Y_k^i = R_i^{\intercal}Y_k \sim N(\0,R_i^{\intercal}\Sigmab_2R_i = \sigma^2_2)$
are one-dimensional. Given $R_i$, we are interested in testing the equality of the two
variances of the projected data. In doing so, the conventional statistic
\begin{align*}
F_i = s^i_1/s^i_2
\end{align*}
is used, where $s^i_1 = \frac{1}{n_1}\sum^{n_1}_{k=1} R_i^{\intercal}X_kX_k^{\intercal}R_i$ and $s^i_2 =\frac{1}{n_2}\sum^{n_2}_{k=1}  R_i^{\intercal}Y_kY_k^{\intercal}R_i$ are the  sample variances of the projected data $\{X^i_k\}$ and $\{Y^i_k\}$, respectively.  For each $R_i$, we have,  under $H^{2}_0$,
\begin{align*}
P(F_i<t) = P( s^i_1/s^i_2\leq t) & = P\left(  \frac{\sum {X^i_k}^2/(n_1\sigma_1^2)}{\sum {Y^i_k}^2/(n_2 \sigma_2^2)}\leq t \right) \\
 				 		& = EP\left(  \frac{\sum {X^i_k}^2/(n_1\sigma_1^2)}{\sum {Y^i_k}^2/(n_2 \sigma_2^2)}\leq t |R_i\right) \\
                                                & = EP(F_{n_1,n_2}<t) \\
  					        & =P(F_{n_1,n_2}<t),
\end{align*}
where $F_{n_1,n_2}$ has a F-distribution with degrees of freedom $n_1$ and $n_2$. Therefore, $F_i$ has a F-distribution with degrees of freedom $n_1$ and $n_2$.
The testing procedure is the same for one-sample and two-sample cases.
To test $H^{2}_0$ based on $m$ projections, the test statistic is given by
\begin{equation}\label{test-2}
 \max_{1\leq i \leq m} F_i.
\end{equation}
A random variable $F$ having F-distribution with degrees of freedom $n_1$ and $n_2$ can be written as
\begin{equation}\label{F}
F = \frac{\frac{U_{n_1}}{n_1}}{\frac{V_{n_2}}{n_2}},
\end{equation}
where $U_{n_1}$ and $V_{n_2}$ are two independent chi-square random variables with degrees of freedom $n_1$ and $n_2$, respectively. To derive the asymptotic normality, we use natural logarithm, $\ln$, function
\begin{equation}\label{log-F}
F^* = \ln(F)\cdot\left( {2/n_1 + 2/n_2} \right)^{-1/2},
\end{equation}
which is commonly known as the variance stabilizing transformation~\cite{Shoemaker-2003}.
Under normal population, one can show that the variance of $\ln(U_{n_1})$ is approximately $2/n_1$ (\cite{Bartlett-Kendall-1946}). It can be shown that the two chi-square random variables  $n_1s^i_1/\sigma_i^2$ and  $n_2s^i_2/\sigma_i^2$ are asymptotically independent in the sense that, after normalization, both statistics are asymptotically normally distributed with covariance equal to zero.
  The transformed $F^*$ is thus asymptotically normally distributed with variance approximately $2/n_1 + 2/n_2$.  Finally, the test statistic
\begin{equation}
T_{2,m} = \max_{1\leq i\leq m} F^*_i,
\end{equation}
is used for testing the null hypothesis $H^{2}_0$ in (\ref{2t-prob}).

\begin{theorem}
 Under $H^2_0$, suppose the common covariance matrix is $\Sigmab$ such that $\tr(\Sigmab)=O(p)$.  As $\min(n_1,n_2) \to \infty$ and $p\to\infty$,
\begin{equation}
T_{2,m} \overset{D}{\to}  \max_{1\leq i \leq m}Z_i,
\end{equation}
where $Z_i's$ are normally distributed with mean $\0$ and $\Cov(Z_i,Z_j) = \Cov(F^*_i ,F^*_j)$.
\end{theorem}

\begin{proof}
The convergence to multivariate normal is straightforward according to multivariate central limit theorem. We show now the covariance
$\Cov(s^i_1/\sigma_1^2,s^i_2/\sigma_2^2)=0.$
It follows from the definition that
\begin{align*}
\Cov(s^i_1/\sigma_1^2,s^i_2/\sigma_2^2) & =   \frac{1}{n_1n_2}\sum_{k}\sum_{h}\Cov\left( \frac{R_i^{\intercal}X_kX_k^{\intercal}R_i}{R_i^{\intercal}\Sigmab R_i} , \frac{R_i^{\intercal}Y_hY_h^{\intercal}R_i}{R_i^{\intercal}\Sigmab R_i}\right).
\end{align*}
Using the  property of conditional expectation, we have
\begin{align*}
&\Cov\left( \frac{R_i^{\intercal}X_kX_k^{\intercal}R_i}{R_i^{\intercal}\Sigmab R_i} , \frac{R_i^{\intercal}Y_hY_h^{\intercal}R_i}{R_i^{\intercal}\Sigmab R_i}\right) \\
&\quad =EE\left( \frac{R_i^{\intercal}X_kX_k^{\intercal}R_i}{R_i^{\intercal}\Sigmab R_i} \cdot \frac{R_i^{\intercal}Y_hY_h^{\intercal}R_i}{R_i^{\intercal}\Sigmab R_i}\Big|R_i\right) - \left\{EE\left( \frac{R_i^{\intercal}X_kX_k^{\intercal}R_i}{R_i^{\intercal}\Sigmab R_i}\Big|R_i\right)\right\}^2 \\
&\quad  = E\left\{  E\left( \frac{R_i^{\intercal}X_kX_k^{\intercal}R_i}{R_i^{\intercal}\Sigmab R_i}\Big|R_i\right) E\left( \frac{R_i^{\intercal}Y_hY_h^{\intercal}R_i}{R_i^{\intercal}\Sigmab R_i}\Big|R_i\right)   \right\} - \left\{ E\left(\frac{\tr(R_iR_i\Sigmab)}{R_i\Sigmab R_i}\right) \right\}^2 \\
&\quad  = E\left\{    \left(\frac{\tr(R_iR_i\Sigmab)}{R_i\Sigmab R_i}\right) \cdot  \left(\frac{\tr(R_iR_i\Sigmab)}{R_i\Sigmab R_i}\right) \right\} -1\\
&\quad  = 1-1=0.
\end{align*}
With the assumption $\tr(\Sigmab) = O(p)$,  we then can apply Lemma~\ref{cov} and \ref{asym-ind} on this two-sample case.
Together with  Lemma~\ref{asym-ind} it follows that $Z_i's$ are asymptotically independent. This completes the proof.
\end{proof}

\begin{remark}
If we are interested in testing whether the difference of two covariance matrices $\Sigmab_1-\Sigmab_2$ is positive (or negative) definite, we may use the test statistic  $\max_{1\leq i\leq m} F^*_i$ (or $ \min_{1\leq i\leq m} F^*_i$). Otherwise, we may use a two-sided test.
\end{remark}

\setcounter{chapter}{4}
\setcounter{equation}{0}
\noindent {\bf 4. Simulations }

Simulations are conducted to evaluate the empirical sizes and powers of the one-sample and two-sample tests based on 1000 replicates. The underlying distribution are assumed to be independent Gaussian distributions with mean 0 and standard deviation 1.
The focus of the simulations study is on the two-sample case where we compare our method with two other well-known methods in the literature.


\noindent {\bf 4.1 One-sample tests }
Table~\ref{T11} gives the sizes of the one-sample test for various $m$ and $p$. It can be seen that the sizes are controlled  well for all $m\leq 1,000$ and for very large $p$ compared to $n$.

\begin{table}[!hb]
\caption{ (Two-sided) Empirical sizes for various $p$ and $m$.   }
\begin{center} \label{T11}
\begin{tabular}{@{}cccccccc@{}}
$n_1=n_2$  &  $p$ & $m=10$   & $m=100$  & $m=1,000$     \\\midrule
50&  32&0.053&0.044&0.051\\
&  64&0.042&0.036&0.061\\
& 128&0.043&0.043&0.059\\
& 256&0.062&0.057&0.063\\
& 512&0.054&0.045& 0.06\\
&1024& 0.04&0.039& 0.05\\
&2048&0.045&0.048&0.041\\
&4096&0.054& 0.05& 0.06\\
100&  32&0.038& 0.05&0.052\\
&  64&0.051&0.053&0.048\\
& 128&0.055&0.054&0.043\\
& 256&0.053&0.048&0.049\\
& 512&0.051&0.049&0.059\\
&1024&0.054&0.047&0.055\\
&2048&0.052&0.051&0.049\\
&4096&0.058&0.045& 0.06\\
200&  32& 0.04&0.043&0.047\\
&  64&0.044& 0.05&0.039\\
& 128&0.044&0.047&0.055\\
& 256&0.041&0.047&0.053\\
& 512&0.054&0.043&0.048\\
&1024& 0.04&0.054&0.062\\
&2048& 0.06&0.039&0.061\\
&4096&0.051&0.045&0.057\\
\hline
\end{tabular}
\end{center}
\end{table}

\noindent {\bf 4.2 Two-sample tests }

In this section, empirical powers are evaluated to show the performance of our method based onhe two scenarios. First, we consider the models (4,.1) and (4.2) in \cite{li-2012} as the null and alternative hypotheses, respectively. Under (4.1),  the first population is set according to
$$X_{ij} = Z_{ij} + \theta_1 Z_{i j+1},$$ and under (4.2) the second population is set according to
$$Y_{ij} = Z_{ij} + \theta_1 Z_{i j+1}+ \theta_2 Z_{ij+2},$$
where $\theta_1=2$ and $\theta_2=1$.
 The second scenario is that two population covariance matrices are of the forms
\begin{align}
\small
\Sigmab_1 = \begin{pmatrix}
d_1 &\rho_1& \cdots& \rho_1^{p-2}& \rho_1^{p-1}  \\
\rho_1 & d_1 & \cdots& \rho_1^{p-3} & \rho_1^{p-2}\\
\vdots & \vdots  & \ddots &  \vdots& \vdots\\
 \rho_1^{p-2} &   \rho_1^{p-3}&  \cdots & d_1 & \rho_1\\
 \rho_1^{p-1} & \rho_1^{p-2} &  \cdots & \rho_1& d_1
 \end{pmatrix}_{p\times p}
\mbox{ and }
 \Sigmab_2 = \begin{pmatrix}
d_2 &\rho_2& \cdots& \rho_2^{p-2}& \rho_2^{p-1}  \\
\rho_2 & d_2 & \cdots& \rho_2^{p-3} & \rho_2^{p-2}\\
\vdots & \vdots  & \ddots &  \vdots& \vdots\\
 \rho_2^{p-2} &   \rho_2^{p-3}&  \cdots & d_2 & \rho_2\\
 \rho_2^{p-1} & \rho_2^{p-2} &  \cdots & \rho_2& d_2
 \end{pmatrix}_{p\times p}.
\end{align}
  We want to test whether these two covariance matrices are equal or not. Note that if $d_1 = d_2$, then the eigenvalues of $\Sigmab_1-\Sigmab_2$ sum to 0 or $\Sigmab_1-\Sigmab_2$ is singular.
We consider three cases such that the covariance matrices are positive definite: (i) $(d_1,\rho_1,d_2,\rho_2) =(1.2, 0.1 , 1 ,0.1)$, (ii) $(d_1,\rho_1,d_2,\rho_2) =( 1.5,0.5,1, 0.6)$ and  (iii) $(d_1,\rho_1,d_2,\rho_2) =( 1.1,0.2,1,0.24)$. In case (i), two covariance matrices differ in the diagonal. The two covariance matrices are entirely different by a smaller amount in case (ii) and by a larger amount in case (iii), indicating that the signals are weak and stronger, respectively.

The cut-off for the upper-tailed test statistics is the $100 \times (1-\alpha) $ percentile of $\max_{i\leq m}{Z_i}$. Since $Z_i's$ are independent identically distributed ( i.i.d.), the cut-off value at the level $\alpha$  can be given explicitly by $z_{1-(1-\alpha)^{1/m}}$. By symmetry, the cut-off value for the lower-tailed test statistic based on  $\min_{i\leq m}{Z_i}$ is equal to $-z_{1-(1-\alpha)^{1/m}}$. Table~\ref{T1} gives the upper-tailed cut-off values for various $m$ at different $\alpha$ levels.  Under the null hypothesis, we assume independent standard Gaussian.
Table~\ref{T1} gives the cut-off values for various $m$ and $\alpha$ values. The empirical sizes of the two-sample test are given in Table~\ref{T3} and it can be seen that the sizes are controlled fairly well for all $m\leq 1,000.$

\begin{table}[!t]
\caption{ The upper-tailed cut-off values  for various $m$ and significance levels $\alpha$.  }
\begin{center}\label{T1}
\begin{tabular}{@{}cccccccc@{}}
\hline$\alpha$  &$m=10$ & $m=100$ &  $m=1,000$        \\\hline
0.01 	&	 3.0889   & 3.7178  & 4.2638  	\\
0.025 &          2.8034   & 3.4774   & 4.0527   \\
0.05  &           2.5679   & 3.2834 &   3.8844  \\
0.1   	&	 2.3087   & 3.0748  &  3.7058  	\\
\hline
\end{tabular}
\end{center}
\end{table}

\begin{table}[!h]
\caption{ (One-sided) Empirical sizes for various $p$ and $m$.   }
\begin{center} \label{T3}
\begin{tabular}{@{}cccccccc@{}}
$n_1=n_2$  &  $p$  &$m=10$  & $m=100$  & $m=1,000$     \\\midrule
50&  32& 0.051& 0.065&0.0737\\
&  64&0.0473&0.0603&0.0713\\
& 128&0.0527& 0.058&0.0813\\
& 256& 0.057&0.0497&0.0827\\
& 512&0.0577&0.0583&0.0727\\
&1024&0.0547& 0.075&0.0737\\
&2048&0.0547&0.0657&0.0733\\
&4096&0.0547&0.0627& 0.068\\
100&  32&0.0497&0.0533&0.0593\\
&  64& 0.047&0.0627&0.0673\\
& 128&0.0517& 0.061&0.0623\\
& 256& 0.047&0.0543&0.0623\\
& 512& 0.061& 0.053&0.0677\\
&1024&0.0507& 0.063&0.0673\\
&2048&0.0513&0.0623&0.0627\\
&4096& 0.052& 0.063& 0.062\\
200&  32&0.0493&0.0537& 0.049\\
&  64&0.0517&0.0487&0.0507\\
& 128& 0.046&0.0497& 0.059\\
& 256&  0.05&0.0523&0.0583\\
& 512& 0.059&0.0507& 0.047\\
&1024&0.0497& 0.055& 0.056\\
&2048&0.0507& 0.061&0.0563\\
&4096& 0.055&0.0517& 0.058\\
\hline
\end{tabular}
\end{center}
\end{table}

 In Table~\ref{T0}, the random projection approach performs worse than CLRT and Li \& Chen's method, while in Tables~\ref{Ti}-\ref{Tiii} our method performs  better for all $m$ even for weak signal. We want to emphasize that the expected value of $\sigma_1-\sigma_2$ under the null hypothesis is equal to the sum of all eigenvalues of $\Sigmab_1 - \Sigmab_2$. Hence, the worst scenario for our method is when the difference of the two covariance matrices is singular, i.e. the sum of all the eigenvalues is equal to 0. In this case, the two projected data sets become indistinguishable in terms of their variances. On the contrary, our method is particularly suitable when the difference between two covariance matrices is a positive (negative) definite matrix. In the first scenario of our simulations, the difference of two covariance matrices is close to be singular and the differences of  $\Sigmab_1$ and $\Sigmab_2$ for the three cases in the second scenario are all positive definite matrices. This explains why our method performs well in the second scenario.

\begin{table}[!t]
\caption{Empirical powers for various $p$ and $m$ with models (4.1) and (4.2) in \cite{li-2012}.   }
\begin{center} \label{T0}
\begin{tabular}{@{}ccccccccccc@{}}
$n_1=n_2$  &  $p$  & CLRT  & Li \& Chen& $m=10$  & $m=100$  & $m=1,000$     \\\midrule
50&  32&1&0.7603&0.2503&0.3563&0.5143\\
&  64&&0.7983& 0.199&0.3827&0.5363\\
& 128&&0.8217& 0.208&0.3777&0.5597\\
& 256&&0.8377& 0.213&0.3673&0.5323\\
& 512&&0.8487&0.2327&0.3783& 0.546\\
&1024&&0.8503&0.2597&0.3717& 0.536\\
&2048&&0.8513&0.2477&0.3933&0.5387\\
&4096&&0.8497& 0.247&  0.38&0.5483\\
100&  32&1&0.9963&0.4567&0.6117&0.8163\\
&  64&1&0.9997& 0.344&0.6207& 0.845\\
& 128&&     1&0.3347& 0.635&0.8323\\
& 256&&0.9997&0.3487&0.6033&0.8303\\
& 512&&     1&0.3753&0.6197&0.8117\\
&1024&&0.9997&0.3697&0.6037&0.8047\\
&2048&&0.9997& 0.384&0.6103& 0.809\\
&4096&&     1&0.4107&0.5937& 0.808\\
200&  32&1&     1&0.7737&0.9227&0.9937\\
&  64&1&     1& 0.579&0.9203&0.9953\\
& 128&1&     1& 0.555& 0.915&0.9957\\
& 256&&     1& 0.553& 0.915& 0.996\\
& 512&&     1&0.5973&  0.92& 0.994\\
&1024&&     1&0.6167&0.9047&0.9963\\
&2048&&     1& 0.625&0.9133& 0.996\\
&4096&&     1&0.6493&0.9133&0.9927\\
\hline
\end{tabular}
\end{center}
\end{table}

 \begin{table}[!t]
\caption{Empirical powers for various $p$ and $m$ with case (i).   }
\begin{center} \label{Ti}
\begin{tabular}{@{}ccccccccccc@{}}
$n_1=n_2$  &  $p$  & CLRT  & Li \& Chen& $m=10$  & $m=100$  & $m=1,000$     \\\midrule
50&  32&0.0750&  0.11& 0.248& 0.356& 0.474\\
&  64&&0.1213& 0.238&0.3657& 0.508\\
& 128&&0.1067&0.2793&0.3753&0.5103\\
& 256&& 0.116& 0.258& 0.373&0.5263\\
& 512&& 0.109&0.2617&0.3727&0.5223\\
&1024&& 0.105&0.2573&0.3877&0.5417\\
&2048&& 0.097& 0.264&0.3827&0.5403\\
&4096&&0.1013&0.2447&0.3687&0.5287\\
100&  32&0.1830&0.2033&  0.37&0.5553& 0.705\\
&  64&0.0960& 0.199&0.3877&0.5787&0.7563\\
& 128&&0.2023&0.3967&0.5907&0.7883\\
& 256&&0.2043& 0.413&0.6113&0.7937\\
& 512&&0.2067&0.4007&0.5983&0.7953\\
&1024&& 0.188&0.3913& 0.602& 0.786\\
&2048&& 0.202&0.4047&0.6153&0.7973\\
&4096&&0.1973& 0.395& 0.596&0.7973\\
200&  32&0.4400& 0.505& 0.609&0.8397&0.9463\\
&  64&0.3140& 0.486&  0.63& 0.882&0.9777\\
& 128&0.2070& 0.474&0.6587&  0.89&0.9863\\
& 256&&0.4927&0.6593& 0.895&0.9887\\
& 512&& 0.477& 0.646&0.9067&0.9903\\
&1024&&0.4807&0.6497&0.8997&0.9933\\
&2048&& 0.476&0.6457&0.9113&0.9923\\
&4096&& 0.487&0.6427&  0.91& 0.992\\
\hline
\end{tabular}
\end{center}
\end{table}

 \begin{table}[!t]
\caption{Empirical powers for various $p$ and $m$ with case (ii).   }
\begin{center} \label{Tii}
\begin{tabular}{@{}ccccccccccc@{}}
$n_1=n_2$  &  $p$  & CLRT  & Li \& Chen& $m=10$  & $m=100$  & $m=1,000$     \\\midrule
50&  32&0.983&0.4773&0.7993&0.9757&0.9983\\
&  64&&  0.49&0.8113&0.9773&0.9997\\
& 128&&0.5017&0.8097&0.9707&     1\\
& 256&&0.4933&0.8283&0.9723&     1\\
& 512&&0.5007&0.7983&0.9737&0.9993\\
&1024&&  0.48& 0.763& 0.966&0.9997\\
&2048&&0.4887& 0.748&0.9633&     1\\
&4096&&0.4893&0.7447&0.9623&0.9993\\
100&  32&1& 0.935&0.9827&     1&     1\\
&  64&1&0.9437&0.9877&     1&     1\\
& 128&&0.9473& 0.984&     1&     1\\
& 256&& 0.946&0.9897&     1&     1\\
& 512&& 0.942&0.9777&     1&     1\\
&1024&& 0.948& 0.973&     1&     1\\
&2048&&0.9437&0.9687&     1&     1\\
&4096&& 0.949&0.9687&     1&     1\\
200&  32&1&     1&     1&     1&     1\\
&  64&1&     1&     1&     1&     1\\
& 128&1&     1&     1&     1&     1\\
& 256&&     1&     1&     1&     1\\
& 512&&     1&     1&     1&     1\\
&1024&&     1&     1&     1&     1\\
&2048&&     1&0.9997&     1&     1\\
&4096&&     1&0.9997&     1&     1\\
\hline
\end{tabular}
\end{center}
\end{table}

 \begin{table}[!t]
\caption{Empirical powers for various $p$ and $m$ with case (iii).   }
\begin{center} \label{Tiii}
\begin{tabular}{@{}ccccccccccc@{}}
$n_1=n_2$  &  $p$  & CLRT  & Li \& Chen& $m=10$  & $m=100$  & $m=1,000$     \\\midrule
50&  32&0.0880&0.0687&0.1373&0.1677&0.2283\\
&  64&&0.0743&0.1397&0.1773&0.2423\\
& 128&&0.0643&0.1403& 0.172&0.2313\\
& 256&& 0.071&0.1313&0.1737&0.2377\\
& 512&&0.0667&0.1407& 0.186& 0.235\\
&1024&& 0.066&0.1307& 0.186&0.2437\\
&2048&&  0.06& 0.144&0.1773&0.2413\\
&4096&& 0.063&  0.14&0.1733& 0.244\\
100&  32&0.1070&0.0927&0.1887& 0.228&0.3013\\
&  64&0.0810& 0.083&0.1953&0.2343&0.3233\\
& 128&&0.0827& 0.174& 0.236&0.3083\\
& 256&& 0.083&0.1723& 0.235& 0.323\\
& 512&&0.0873&0.1643& 0.232& 0.307\\
&1024&&  0.08&0.1753& 0.236& 0.306\\
&2048&& 0.089&0.1817& 0.239& 0.311\\
&4096&&0.0813&0.1813&0.2433&0.3113\\
200&  32&0.2650&0.1617&0.3013&0.3787&0.4847\\
&  64&0.2040&0.1407&0.2897& 0.381&0.4947\\
& 128&0.1330&0.1327& 0.279&0.3803&  0.52\\
& 256&& 0.152&0.2567& 0.377& 0.499\\
& 512&&0.1457&0.2613& 0.382&0.4947\\
&1024&& 0.137&0.2377&0.3723& 0.503\\
&2048&&0.1453& 0.259& 0.358&0.4943\\
&4096&& 0.154&0.2513& 0.359& 0.502\\
\hline
\end{tabular}
\end{center}
\end{table}

\vspace{0.1in}
\setcounter{chapter}{5}
\setcounter{equation}{0}
\noindent {\bf 5. Summary and discussion}

We present a simple method to test the high-dimensional covariance matrices for both one-sample and two-sample cases using random projection. In general, a random projection method is to project data from $\mathbb{R}^p$ to $\mathbb{R}^k$, $1\leq k<n$. In this paper, we focus on $k=1.$ The reason we opt to use $k=1$ is that the random projection method only processes the one-dimensional data so that it is very efficient in computation and easy to use.  In addition, the interpretation is simple. We then illustrate our method based on the likelihood ratio test statistics and derive the asymptotic normality for the null distributions.  The asymptotic results hold when both $n$ and $p$ go to infinity, and there is no relationship required between $n$ and $p$. However, by adding some minor conditions, we can obtain the approximate convergence rate of the covariance between different random projections. Finally, simulations are conducted to compare our method with \cite{li-2012}'s method and CLRT introduced by \cite{Bai-et-al-2009}. Surprisingly, our method performs very well in a certain class of covariance matrices. The derivation and  numerical results show that the random projection method is advantageous when the difference between two covariance is almost positive definite (negative) and disadvantageous when the difference is almost singular. By almost positive (negative) definite we mean most of the eigenvalues are positive (negative) or the sum of the eigenvalues is large (small), and by almost singular we mean the sum of the eigenvalues is close to zero. If the difference of two covariance matrices is almost positive (negative) definite, then the strength of the signal (difference) is well-preserved by random projection onto the one-dimensional space. If the difference of two covariance matrices is singular, then the signal is then completely masked by random projection. For example, our method is not suitable for testing two covariance matrices having the same diagonal. In such a case, no matter how large the signals are on the off-diagonal entries, the difference of the two variances of the projected data is zero on average.



\vskip 14pt
\noindent {\large\bf Acknowledgement}\\
 
The work was partially supported by NSF-III-1360971, NSF-Bigdata-1419210, ONRN00014-13-1-0764, and AFOSR-FA9550-13-1-0137. The work was completed while the the first author was a postdoctoral researcher at Rutgers University in 2013-2014.

\par

\markboth{\hfill{\footnotesize\rm Tung-Lung Wu AND Ping Li} \hfill}
{\hfill {\footnotesize\rm  High-Dimensional Covariance Matrices} \hfill}


\vskip .65cm
\noindent
Rutgers University
\vskip 2pt
\noindent
E-mail: (wutunglung@gmail.com)
\vskip 2pt

\noindent
Rutgers University
\vskip 2pt
\noindent
E-mail: (pingli@stat.rutgers.edu)
\end{document}